\documentclass[runningheads]{llncs}
\usepackage{hyperref}
\usepackage{latexsym,amssymb} 
\usepackage{color} 
\usepackage[only,llbracket,rrbracket,llbrace,rrbrace,llparenthesis,rrparenthesis,lightning,Lbag,Rbag,%
ovee]{stmaryrd}  

\usepackage{tikz-cd}    %For diagram.
\usepackage{subcaption} %For subfigure
\newcommand{\point}{{\bullet}} % point in picture
\tikzcdset{row sep/normal=1ex}

\renewcommand{\phi}{\varphi}

\RequirePackage{scalerel}
\RequirePackage{accsupp}
\newcommand*{\llbrace}{%
  \BeginAccSupp{method=hex,unicode,ActualText=2983}%
    \textnormal{\usefont{OMS}{lmr}{m}{n}\char102}%
    \mathchoice{\mkern-4.05mu}{\mkern-4.05mu}{\mkern-4.3mu}{\mkern-4.8mu}%
    \textnormal{\usefont{OMS}{lmr}{m}{n}\char106}%
  \EndAccSupp{}%
}
\newcommand*{\rrbrace}{%
  \BeginAccSupp{method=hex,unicode,ActualText=2984}%
    \textnormal{\usefont{OMS}{lmr}{m}{n}\char106}%
    \mathchoice{\mkern-4.05mu}{\mkern-4.05mu}{\mkern-4.3mu}{\mkern-4.8mu}%
    \textnormal{\usefont{OMS}{lmr}{m}{n}\char103}%
  \EndAccSupp{}%
}

\newcommand{\dt}[1]{\textbf{\emph{#1}}} % defined term

\renewcommand{\c}{\mathbin{\mbox{\small\upshape\textbf{;}}}} %composition - from cmho
\newcommand{\lagree}{\mathring{\sim}}
\newcommand{\ssc}{\mathsf{ssc}\,}
\newcommand{\PSC}{\mathsf{PSC}\,}
\newcommand{\IPF}{\mathsf{IPF}}

\newcommand{\dirimg}[1]{\langle #1 \rangle} % direct image (Note: mathtools is a nice package that gives a \DeclarePairedDelimiter command, would make this nicer)
\newcommand{\lfp}{\mathsf{lfp}\,}
\newcommand{\Ni}{\mathord{\ni}} % epsiloff
\newcommand{\join}{\mathbin{\sqcup}} % pointwise join (lifted union) at trans and at hyper level 
\newcommand{\ijoin}{\mathbin{\ovee}} % inner join
\newcommand{\choice}{\oplus} % choice in programs 
\newcommand{\pbot}{\underline{\bot}} % powerset bottom

\newcommand{\converse}[1]{#1^\cup} % relational converse

\newcommand{\sbs}{\subseteq}
\newcommand{\sps}{\supseteq}
\newcommand{\rsps}{\mathord{\supseteq}} % \sps as relation

\newcommand{\hycond}[1]{\mathbin{ \triangleleft~#1~\triangleright }} % infix hyper-conditional 
\newcommand{\relto}{\mathbin{\multimap}} % binary relations 
\newcommand{\Sta}{\Sigma}
\newcommand{\Trc}{\mathsf{Trc}}

\newcommand{\rdom}{\mathsf{dom}} % domain of relation
\newcommand{\Dom}{\mathsf{Dom}} % domain of transformer

\newcommand{\aftqua}{\:\mbox{\small$\bullet$}\:}
\newcommand{\all}[2]{\forall #1 \aftqua #2}
\newcommand{\some}[2]{\exists #1 \aftqua #2}

\newcommand{\pset}{\wp} % powerset 
\newcommand{\ppset}{\pset^2} % powerset of powerset
\newcommand{\Dpset}{\breve{\pset}} % subset closed ppset

\newcommand{\union}{\mathbin{\mbox{\small$\cup$}}}
\newcommand{\intersect}{\mathbin{\mbox{\small$\cap$}}}
\newcommand{\imp}{\Rightarrow}
\newcommand{\impby}{\Leftarrow}

\newcommand{\rt}{\sqsubseteq}
\newcommand{\rf}{\sqsupseteq}

\newcommand{\meansR}[1]{\llbracket\, #1 \,\rrbracket} % relation semantics
\newcommand{\meansT}[1]{\llbrace\, #1 \,\rrbrace} % transformer (forward) semantics
\newcommand{\meansH}[1]{\llparenthesis\, #1 \,\rrparenthesis} % hyper semantics 

\newcommand{\args}{{\!-\!}} % empty args of \means functions

\newcommand{\F}{\mathsf{F}}
\newcommand{\G}{\mathsf{G}}
\renewcommand{\H}{\mathsf{H}}
\newcommand{\K}{\mathsf{K}}
\newcommand{\NI}{\mathsf{NI}}
\newcommand{\Agrl}{\mathsf{Agrl}\,} % agree on low 

\renewcommand{\P}{\mathbb{P}} % set of sets of states
\newcommand{\Q}{\mathbb{Q}} % set of sets of states
\newcommand{\A}{\mathbb{A}} % fixpoints of \Agrl
\newcommand{\Hp}{\mathbb{H}}

\newcommand{\keyw}[1]{\ensuremath{\mathsf{#1}}}

\newcommand{\skipc}{\keyw{skip}}
 
\newcommand{\ifc}[3]{\keyw{if}\ {#1}\ \keyw{then}\ {#2}\ \keyw{else}\ {#3}}
\newcommand{\whilec}[2]{\keyw{while}\ {#1}\ \keyw{do}\ {#2}}

\newenvironment{calcu}{                   
  \[ \begin{array}{rl}
  }{
  \end{array} \] }
\newcommand{ \fln }[1]{ & #1             \protect\\*[.6ex] }
\newcommand{ \pln }[3]{%
   #1  &\quad\; \Lbag\mbox{ #3 }\Rbag \protect\\*[.6ex]
       & #2         \protect\\*[.6ex] }
\newenvironment{calcushort}{                   
  \[ \begin{array}{rll}
  }{
  \end{array} \] }
\newcommand{ \flns }[1]{ & #1             \protect\\*[-.5ex] }
\newcommand{ \plns }[3]{%
   #1  &&       \quad \mbox{ #3 } \protect\\*[-.5ex]
       & #2         \protect\\*[-.5ex] }

\title{Whither Specifications as Programs
\thanks{The authors were partially supported by NSF award 1718713}
}
\author{David A. Naumann and Minh Ngo}
\institute{Stevens Institute of Technology, USA\\
\email{naumann@cs.stevens.edu}
\email{nngo1@stevens.edu}
}
\begin{document}
\maketitle

% Kent Beck guideline:
% The first states the problem. The second states why the problem is a problem. The third is my startling sentence. The fourth states the implication of my startling sentence. 

\begin{abstract}
Unifying theories distil common features of programming languages and design methods by means of algebraic operators and their laws.  Several practical concerns --- e.g., improvement of a program, conformance of code with design, correctness with respect to specified requirements  --- are subsumed by the beautiful notion that programs and designs are special forms of specification and their relationships are instances of logical implication between specifications.   Mathematical development of this idea has been fruitful but limited to an impoverished notion of specification: trace properties.  Some mathematically precise properties of programs, dubbed hyperproperties, refer to traces collectively.  For example, confidentiality involves knowledge of possible traces. This article reports on both obvious and surprising results about lifting algebras of programming to hyperproperties, especially in connection with loops, and suggests directions for further research.  The technical results are: a compositional semantics, at the hyper level, of imperative programs with loops, and proof that this semantics coincides with the direct image of a standard semantics, for subset closed hyperproperties.
\end{abstract}

\section{Introduction}\label{sec:intro}

A book has proper spelling provided that each of its sentences does.  
For a book to be captivating and suspenseful --- that is not a property that can be reduced to a property of its individual sentences.  Indeed, few interesting properties of a book are simply a property of all its sentences.
By contrast, many interesting requirements of a program can be specified as so-called \dt{trace properties}: there is some property of traces (i.e., observable behaviors) which must be satisfied by all the program's traces.  

The unruly mess of contemporary programming languages, design tools, and approaches to formal specification has been given a scientific basis through unifying theories that abstract commonalities by means of algebraic operators and laws. 
Algebra abstracts from computational notions like partiality and nondeterminacy by means of operators that are interpreted as total functions and which enable equational reasoning.
Several practical concerns --- such as improving a program's resource usage while not altering its observable behavior, checking conformance of code with design architecture, checking satisfaction of requirements, 
and equivalence of two differently presented designs --- are subsumed by the beautiful notion that programs and designs\footnote{This paper was written with the UTP~\cite{HoareHe} community in mind,
but our use of the term ``design'' is informal and does not refer to the technical notion in UTP.}
are just kinds of specification and their relationships are instances of logical implication between specifications.  
Transitivity of implication yields the primary relationship: the traces of a program are included in the traces allowed by its specification.
The mathematical development of this idea has been very successful --- for trace properties.

Not all requirements are trace properties.  A program should be easy to read, 
consistent with dictates of style, and amenable to revision for adapting to changed requirements.
Some though not all such requirements may be addressed by mathematics; e.g., parametric polymorphism is a form of modularity that facilitates revision through reuse.  In this paper we are concerned with requirements that are extensional in the sense that they pertain directly to observable behavior.  
For a simple example, consider a program acting on variables $hi,lo$ where the initial value of $hi$ is meant to be a secret, on which the final value of $lo$ must not depend.  Consider this simple notion of program behavior: 
a state assigns values to variables, and a trace is a pair: the initial and final states.  
The requirement cannot be specified as a trace property, but it can be specified as follows: 
for any two traces $(\sigma,\sigma')$ and $(\tau,\tau')$, if the initial states $\sigma$ and $\tau$ have the same value for $lo$ then so do the final states.
In symbols: $\sigma(lo)=\tau(lo) \imp \sigma'(lo)=\tau'(lo)$.  

Some requirements involve more than two traces, e.g., ``the average response time is under a millisecond'' can be made precise by averaging the response time of each trace, over all traces, perhaps weighted by a distribution that represents likelihood of different requests.
For a non-quantitative example, consider the requirement that a process in a distributed system should know which process is the leader: something is known in a given execution if it is true in all possible traces that are consistent with what the process can observe of the given execution (such as a subset of the network messages).
In the security literature, some information flow properties are defined by closure conditions on the program's traces, such as: for any two traces, there exists a trace with the high (confidential) events of the first and the low (public) events of the second.

% omittable sentence 
This paper explores the notion that just as a property of books is a set of books, not necessarily defined simply in terms of their sentences,
so too a property of programs is a set of programs, not merely a set of traces.  
The goal is to investigate how the algebra of programming can be adapted for reasoning about non-trace properties.
To this end, we focus on the most rudimentary notion of trace, i.e., pre/post pairs, and rudimentary program constructs.
We conjecture that the phenomena and ideas are relevant to a range of models, perhaps even the rich notions of trace abstracted by variations on concurrent Kleene algebra~\cite{HoareMSW11}.

It is unfortunate that the importance of trace properties in programming has led to well established use of the term ``property'' for trace property,
and recent escalation in terminology to ``hyperproperty'' to designate the general notion of program property --- sets of programs rather than sets of traces~\cite{ClarksonS10,ClarksonFKMRS14}.  
Some distinction is needed, so for clarity and succinctness we follow the crowd.  
The technical contribution of this paper can now be described as follows: we give a lifting of the fixpoint semantics of loops to the ``hyper level'', and show anomalies that occur with other liftings.  
This enables reasoning at the hyper level with usual fixpoint laws for loops,
while retaining consistency with standard relational semantics. 
Rather than working directly with sets of trace sets, our lifting uses a simpler model,
sets of state sets; this serves to illustrate the issues and make connections with other models that may be familiar.
The conceptual contribution of the paper is to call attention to the challenge of unifying theories of programming that encompass requirements beyond trace properties.

\paragraph{Outline.}
Section~\ref{sec:rel} describes a relational semantics of imperative programs and defines an example program property that is not a trace property.
Relational semantics is connected, in Section~\ref{sec:trans}, with semantics mapping sets to sets, like forward predicate transformers.
Section~\ref{sec:htrans} considers semantics mapping sets of sets to the same, this being the level at which hyperproperties are formulated.  Anomalies with obvious definitions motivate a more nuanced semantics of loops.  
The main technical result of the paper is Theorem~\ref{thm:main} in this section, connecting the semantics of 
Section~\ref{sec:htrans} with that of Section~\ref{sec:trans}.
Section~\ref{sec:spec} connects the preceding results with the intrinsic notion of satisfaction for hyperproperties,  
and sketches challenges in realizing the dream of reasoning about hyperproperties using only refinement chains.
The semantics and theorem are new, but similar to results in prior work discussed in Section~\ref{sec:related}.
Section~\ref{sec:concl} concludes.

%% IDENTIFIERS
%% $\Sta$ - set of states; $A,B,C$ generic sets
%% $p,q,r$ - set of states
%% $R,S,T$ - relation on states 
%% $\phi,\psi$ - functions on state sets; fits with $\alpha,\gamma$?
%% $\P,\Q$ - set of state sets 
%% $\Phi,\Psi$ - functions on sets of sets of states 

\section{Programs and specifications as binary relations}\label{sec:rel}

\paragraph{Preliminaries.}

We review some standard notions, to fix notation and set the stage.
Throughout the paper we assume $\Sta$ is a nonempty set, which stands for the set of program states, or data values, on which programs act.  
For any sets $A$, $B$,
let $A \relto B$ denote the binary relations from $A$ to $B$; that is,
$A \relto B$ is $\pset(A\times B)$ where $\pset$ means powerset.
Unless otherwise mentioned, we consider powersets, including $\Sta\relto\Sta$, to be ordered by inclusion ($\sbs$).

We write $A\to B$ for the set of functions from  $A$ to $B$.
For composition of relations, and in particular composition of functions, we use infix symbol $\c$ in the forward direction.  
Thus for relations $R,S$ and elements $x,y$ we have
$x(R\c S)y$ iff $\some{z}{xRz \land zSy}$.  
For a function $f: A\to B$ and element $x\in A$ we write application as $f x$ and let it associate to the left.
Composition with $g:B\to C$ is written $f\c g$, as functions are treated as special relations,  so $(f\c g)x = g(f x)$.  
The symbol $\c$ binds tighter than $\union$ and other operators.

For a relation $R:A \relto B$, the direct image
 $\dirimg{R}$ is a total function $\pset A \to \pset B$
defined by
$y\in\dirimg{R}p$ iff $\some{x\in p}{xRy}$.
It faithfully reflects ordering of relations:
\[ %begin{equation}\label{eq:imgInj}
R\sbs S \quad\mbox{iff}\quad \dirimg{R}\rt\dirimg{S}
\]
where $\rt$ means pointwise order % for functions of type $\pset A\to\pset B$
 (i.e., 
$\phi\rt\psi$ iff $\all{ p\in\pset A }{ \phi\, p \sbs \psi\, p}$).
We write $\join$ for pointwise union,
defined by $(\phi\join\psi)\,p = \phi\, p \union \psi\, p$.
The $\rt$-least element is the function $\lambda p \aftqua \emptyset$, abbreviated as $\bot$.
A relation can be recovered from its direct image:
\begin{equation}\label{eq:relRecover}
R = sglt \c \dirimg{R} \c \Ni 
\end{equation}
where $sglt : A \to \pset A$ maps element $a$ to singleton set $\{a\}$ and $\Ni: \pset A \relto A$ is the converse of the membership relation.
Note that $\bot$ is the direct image of the empty relation.
Direct image is functorial and distributes over union:
\[
\dirimg{id_\Sta} = id_{\pset\Sta}
\qquad
\dirimg{R \c S} = \dirimg{R} \c \dirimg{S}
%\qquad R\sbs S \imp \dirimg{R} \rt \dirimg{S} % DN said earlier 
\qquad
\dirimg{R \union S} = \dirimg{R} \join \dirimg{S}
\]
We write $id$ for identity function on the set indicated.
In fact $\dirimg{-}$ distributes over arbitrary union,
i.e., sends any union of relations to the pointwise join of their images.  
Also, $\dirimg{R}$ is universally disjunctive, and (\ref{eq:relRecover}) forms a bijection
between universally disjunctive functions 
$\pset A\to\pset B$ and relations $A\relto B$.

In this paper we use the term \dt{transformer} for monotonic functions of type 
$\pset A \to\pset B$. % (for various sets $A,B$).
For $\phi:\pset A\to\pset B$ to be monotonic is equivalent to 
$(\rsps\c \phi) \sbs (\phi\c\rsps)$.

% FUTURE
%\dn{we can lift this to h-transformers, but be careful when restricting $\rt$ to subset closed}

We write $\lfp$ for the least-fixpoint operator.    
%TODO we may write, e.g., $\lfp^\sbs_\emptyset$ as a hint which order relation and least element is meant.
% But it's not too helpful since we always iterate from bottom in this paper.
For monotonic functions $f:A\to A$ and $g:B\to B$ where $A,B$ are sufficiently complete 
posets that $\lfp f$ and $\lfp g$ exist, 
the \dt{fixpoint fusion} rule says that for strict and continuous $h:A\to B$,
\begin{equation}\label{eq:fusion}
f\c h = h\c g \imp h(\lfp f) = \lfp g
\end{equation}
Inequational forms, such as $f\c h \leq h\c g \imp h(\lfp f) \leq \lfp g$,
are also important.\footnote{Fusion rules, also called fixpoint transfer, 
can be found in many sources, 
e.g.,~\cite{AartsEtal95,Back:book}.
We need the form in Theorem 3 of~\cite{Cousot02},
for Kleene approximation of fixpoints.
}

\paragraph{Relational semantics.}

The relational model suffices for reasoning about terminating executions.  
If we write $x+2\leq x'$ to specify a program that increases $x$ by at least two,
we can write this simple refinement chain:
\[ x+2\leq x' \quad \sps \quad x := x+3 \choice x := x+5 
\quad \sps \quad x := x+3 
\]
to express that the nondeterministic choice ($\choice$) between adding 3 or adding 5 refines the specification and is refined in turn by the first alternative.    
Relations model a good range of operations including relational converse and intersection which are not implementable in general but are useful for expressing specifications.  
Their algebraic laws facilitate reasoning.  
For example, choice is modeled as union, so the second step is from a law of set theory: $R\union S\sps R$.

Equations and inequations may serve as specifications.
For example, to express that relation $R$ is deterministic we can write $\converse{R}\c R \sbs id$,
where $\converse{R}$ is the converse of $R$.
Note that this uses two occurrences of $R$.
%An alternative is $\some{f}{f\sps R}$ where $f$ ranges over partial functions; it is hardly an improvement.
Returning to the example in the introduction, suppose $R$ relates 
%pairs $(h,l)$ of values of 
states with variables $hi,lo$.
To formulate the \dt{noninterference} property that the final value of $lo$ is independent of the initial value of $hi$,
it is convenient to define a relation on states that says they have the same value for $lo$:
define $\lagree$ by $\sigma\lagree\tau$ iff $\sigma(lo)=\tau(lo)$.
The property is 
\[ \all{\sigma,\sigma',\tau,\tau'}{ \sigma R \sigma' \land \tau R \tau'  \land \sigma\lagree\tau  \imp \sigma'\lagree\tau' }\]
This is a form of determinacy. A weaker notion allows multiple outcomes for $lo$ but the set of possibilities should be independent from the initial value of $hi$.
\[ \all{\sigma,\sigma',\tau}{ \sigma R \sigma' \land \sigma\lagree\tau  \imp 
     \some{\tau'}{ \tau R \tau' \land \sigma'\lagree\tau' }}\]
This is known as \dt{possibilistic noninterference}.
It can be expressed without quantifiers, by the usual simulation inequality:
\begin{equation}\label{eq:sim}
\mathord{\lagree}\c R \;\sbs\; R \c \mathord{\lagree} 
\end{equation}
Another equivalent form is
\( \lagree \c R \c \lagree \; = \; R \c \lagree \),
which again uses two occurrences of $R$.
The algebraic formulations are attractive, but recall the beautiful idea of correctness proof as a chain of refinements 
\[ spec \sps design \sps \ldots \sps prog  \]
This requires the specification to itself be a term in the algebra, rather than an (in)equation between terms.

Before proceeding to investigate this issue, we recall the well known fact that 
possibilistic noninterference is not closed under refinement of trace sets~\cite{Jacob88}.
Consider $hi,lo$ ranging over bits, so we can write pairs compactly, and consider the set of traces 
\( \{ 
(00,00), \underline{(00,01)}, 
(01,00), \underline{(01,01)}, 
\underline{(10,10)}, (10,11), 
\underline{(11,10)}, (11,11) 
\} \)
It satisfies possibilistic noninterference, but if we remove the underlined pairs the result does not; in fact
the result copies $hi$ to $lo$.

In the rest of this paper, we focus on deterministic noninterference, $\NI$ for short.
It has been advocated as a good notion for security~\cite{RoscoeWW1994} % 
and it serves our purposes as an example.

\paragraph{A signature and its relational model.}

% NOTE: regular ; for cmds but \c elsewhere

To investigate how $\NI$ and other non-trace properties may be expressed and used in refinement chains,
it is convenient to focus on a specific signature, the simple imperative language over given atoms (ranged over by $atm$) and boolean expressions (ranged over by $b$).
\begin{equation}\label{eq:signature}
c ::= atm \mid \skipc \mid c;c \mid c \choice c \mid \ifc{b}{c}{c} \mid \whilec{b}{c} 
\end{equation}
For expository purposes we refrain from decomposing the conditional and iteration constructs in terms of choice ($\choice$) and assertions.
That decomposition would be preferred in a more thorough investigation of algebraic laws, and it is evident in the semantic definitions to follow.

Assume that for each $atm$ is given a relation $\meansR{atm}:\Sta\relto\Sta$,
and for each boolean expression $b$ is given a coreflexive relation 
$\meansR{b}:\Sta\relto\Sta$.  That is, $\meansR{b}$ is a subset of the identity relation $id_\Sta$ on $\Sta$.
For non-atom commands $c$ the relational semantics $\meansR{c}$ is defined in Fig.~\ref{fig:relsem}.
The fixpoint for loops\footnote{\label{fn:tailrec}
  It is well known that loops are expressible in terms of recursion:
  $\whilec{b}{c}$ can be expressed as $\mu X . (b;c;X \union \neg b)$
  and this is the form we use in semantics.  
  A well known law is $\mu X . (b;c;X \union \neg b) =  \mu X . (b;c;X \union \skipc);\neg b$
  which factors out the termination condition.}
is in $\Sta\relto\Sta$, ordered by $\sbs$ with least element $\emptyset$.

\begin{figure}[t]
\begin{normalsize}
\[
\begin{array}{lcl}
\meansR{\skipc}          &\; = \;& id_\Sta \\
\meansR{c\, ; d}           & = & \meansR{c} \c \meansR{d} \\
\meansR{c \choice d}           & = & \meansR{c} \union \meansR{d} \\
\meansR{ \ifc{b}{c}{d} } & = &  \meansR{b} \c \meansR{c} \:\union\: \meansR{\neg b} \c \meansR{d} \\
\meansR{ \whilec{b}{c} } & = & \lfp \F \\ % \lfp_\emptyset^\sbs \F \\
               &\multicolumn{2}{l}{ 
                 \mbox{where } \F: (\Sta\relto\Sta)\to(\Sta\relto\Sta) \mbox{ is defined }} \\
               &\multicolumn{2}{l}{ \F R = \meansR{b} \c \meansR{c} \c R \:\union\: \meansR{\neg b} }
\end{array}
\]
\end{normalsize}
\caption{Relational semantics $\meansR{c} \in \Sta\relto\Sta$,
with $\meansR{atm}$ assumed to be given.
}
\label{fig:relsem}
\end{figure}

The language goes beyond ordinary programs, in the sense that atoms are allowed to be unboundedly nondeterministic.  They are also allowed to be partial; coreflexive atoms serve as assume and assert statements.
Other ingredients are needed for a full calculus of specifications, but here our aim is to sketch ideas that merit elaboration in a more comprehensive theory.  

\section{Programs as forward predicate transformers}\label{sec:trans}

Here is yet another way to specify $\NI$ for a relation $R$:
\[ \all{p\in\pset\Sta}{ \Agrl(p) \imp \Agrl(\dirimg{R}p) } \]
where $\Agrl$ says that all elements of $p$ agree on $lo$:
\[ 
\Agrl(p) \quad\mbox{iff}\quad \all{\sigma,\tau}{\sigma\in p \land \tau\in p \imp \sigma\lagree\tau} 
\] 
As with the preceding (in)equational formulations, like (\ref{eq:sim}), this is not directly applicable as the specification in a refinement chain, but it does hint that escalating to sets of states may be helpful.
Note that $R$ occurs just once in the condition.

Weakest-precondition predicate transformers are a good model for programming algebra:
Monotonic functions $\pset\Sta\to\pset\Sta$ can model total correctness specifications with both angelic and demonic nondeterminacy. 
In this paper we use transformers to model programs in the forward direction.

For boolean expression $b$ we define 
\( \meansT{b} = \dirimg{ \meansR{b} } \)
so that $\meansT{b}$ is a filter:
$x$ is in $\meansT{b} p$ iff $x\in p$ and $b$ is true of $x$.
The transformer semantics is in Fig.~\ref{fig:transem}.
For loops, the fixpoint is for the aforementioned $\bot$ and $\rt$.

\begin{figure}[t]
\begin{normalsize}
\[
\begin{array}{lcl}
\meansT{atm}             &\; = \;& \dirimg{\meansR{atm}} \\
\meansT{\skipc}          & = & id_{\pset\Sta} \\
\meansT{c\, ; d}           & = & \meansT{c} \c \meansT{d} \\
\meansT{c \choice d}      & = & \meansT{c} \join \meansT{d} \\
\meansT{ \ifc{b}{c}{d} } & = &  \meansT{b} \c \meansT{c} \:\join\: \meansT{\neg b} \c \meansT{d} \\
\meansT{ \whilec{b}{c} } & = & \lfp \G  \\ % \lfp_\bot^{\rt} \G  \\
               &\multicolumn{2}{l}{ 
         \mbox{where } \G: (\pset\Sta\to\pset\Sta)\to(\pset\Sta\to\pset\Sta) \mbox{ is defined }} \\
               &\multicolumn{2}{l}{ 
                         \G \phi = \meansT{b}\c\meansT{c}\c\phi \join \meansT{\neg b} }
\end{array}
\]
\end{normalsize}
\caption{Transformer semantics  $\protect\meansT{c} \in \pset\Sta\to\pset\Sta$.
}\label{fig:transem}
\end{figure}

\paragraph{Linking transformer with relational.}

The transformer model may support a richer range of operators than the relational one,
but for several reasons it is important to establish their mutual consistency on a common set of operators~\cite{HoareL74,HoareHe}.
A relation can be recovered from its direct image, see (\ref{eq:relRecover}), so the following is a strong link.

\begin{proposition}\label{prop:reltran}
For all $c$ in the signature, $\dirimg{\meansR{c}} = \meansT{c}$.
\end{proposition}
\begin{proof}
By induction on $c$. 
\begin{itemize}
\item
$\skipc$: $\dirimg{\meansR{\skipc}} = \dirimg{id_{\Sta}} = id_{\pset\Sta\to\pset\Sta} = \meansT{\skipc} $
by definitions and $\dirimg{-}$ law.

\item
$atm$: $\dirimg{\meansR{atm}} = \meansT{atm}$ by definition.

\item
$c ; d$: $\dirimg{\meansR{c; d}}
= \dirimg{\meansR{c}\c\meansR{d}}
= \dirimg{\meansR{c}}\c\dirimg{\meansR{d}}
= \meansT{c}\c\meansT{d}
= \meansT{c ; d}$
by definitions, $\dirimg{-}$ laws, and induction hypothesis.

\item
$c\choice d$: $\dirimg{\meansR{c\choice d}}
= \dirimg{\meansR{c}\union\meansR{d}}
= \dirimg{\meansR{c}}\join\dirimg{\meansR{d}}
= \meansT{c}\join\meansT{d}
= \meansT{c\choice d}$ 
by definitions, $\dirimg{-}$ laws, and induction hypothesis.

\item
$\ifc{b}{c}{d}$:
$\dirimg{\meansR{\ifc{b}{c}{d}}}
= \dirimg{\meansR{b}\c\meansR{c}\union\meansR{\neg b}\c\meansR{d}}
= \dirimg{\meansR{b}}\c\dirimg{\meansR{c}}\join\dirimg{\meansR{\neg b}}\c\dirimg{\meansR{d}}
= \meansT{b}\c\dirimg{\meansR{c}}\join\meansT{\neg b}\c\dirimg{\meansR{d}}
= \meansT{b}\c\meansT{c}\join\meansT{\neg b}\c\meansT{d}
= \meansT{\ifc{b}{c}{d}}
$
by definitions, $\dirimg{-}$ laws, and induction hypothesis.
\item
$\whilec{b}{c}$:  To prove 
$\dirimg{\meansR{\whilec{b}{c}}} = \meansT{\whilec{b}{c}}$,
unfold the definitions to 
$\dirimg{\lfp \F} = \lfp \G$, where $\F,\G$ are defined in Figs~\ref{fig:relsem} and~\ref{fig:transem}.
This follows by fixpoint fusion, taking $h$ in  (\ref{eq:fusion}) to be $\dirimg{-}$
so the antecedent to be proved is 
$\all{R}{ \dirimg{ \F R } = \G\dirimg{R} }$.  
Observe for any $R$:
\begin{calcushort}
\flns{    \dirimg{\F R} }
\plns{=}{ \dirimg{ \meansR{b} \c \meansR{c} \c R \union \meansR{\neg b}} }{def $\F$}

%\plns{=}{ \dirimg{ \meansR{b} \c \meansR{c} \c R } \join \dirimg{\meansR{\neg b}} }{$\dirimg{-}$ dist $\union$}

\plns{=}{ \dirimg{ \meansR{b}} \c \dirimg{\meansR{c}} \c \dirimg{R } \join \dirimg{\meansR{\neg b}} }{$\dirimg{-}$ distributes over $\c$ and $\union$}

\plns{=}{ \meansT{b} \c \dirimg{\meansR{c}} \c \dirimg{R } \join \meansT{\neg b} }{def $\meansT{b}$}

\plns{=}{ \meansT{b} \c \meansT{c} \c \dirimg{R} \join \meansT{\neg b} }{induction hypothesis}

\plns{=}{ \G\dirimg{R} }{def $\G$}
\end{calcushort}
\end{itemize}
\qed
\end{proof}
Subsets of $\pset \Sta \to \pset \Sta$, such as transformers satisfying  Dijkstra's healthiness conditions,
validate stronger laws than the full set of (monotonic) transformers.
Healthiness conditions can be expressed by inequations,
such as the determinacy inequation $\converse{R}\c R \sbs id$, 
and used as antecedents in algebraic laws.
Care must be taken with joins: not all subsets are closed under pointwise union.
Pointwise union does provide joins in the set of all transformers and also in the set of all universally disjunctive transformers.

In addition to transformers as weakest preconditions~\cite{Back:book}, 
another similar model is multirelations which are attractive in maintaining a pre-to-post direction~\cite{MartinCR07}.
These are all limited to trace properties, though, so we proceed in a different direction.

\section{Programs as h-transformers}\label{sec:htrans}

Given $R:A\relto B$, the image $\dirimg{R}$ is a function and functions are relations, so the direct image can be taken:
\( \dirimg{\dirimg{R}} : \ppset A \to \ppset B \)
where $\ppset A$ abbreviates $\pset(\pset A)$.
In this paper, monotonic functions of this type are called \dt{h-transformers}, in a nod to hyper terminology.

The underlying relation can be recovered by two applications of (\ref{eq:relRecover}):
\[ R = sglt \c sglt \c \dirimg{\dirimg{R}} \c \Ni \c \Ni \]
More to the point, a quantifier-free  formulation of $\NI$ is now in reach.
Recall that we have $R \in \NI$ iff 
$\all{p\in\pset\Sta}{ \Agrl(p) \imp \Agrl(\dirimg{R} p) }$.
This is equivalent to 
\begin{equation}\label{eq:NI}
\dirimg{\dirimg{R}} \A \sbs \A
\end{equation}
where the set of sets $\A$ is defined by $\A = \{ p \mid \Agrl(p) \}$.
This is one motivation to investigate $\ppset \Sta \to \ppset \Sta$ as a model, rather than 
$\pset(\Sta\relto\Sta)$ which is the obvious way to embody the idea that a program is a trace set and a property is a set of programs.

In the following we continue to write $\join$ and $\rt$ for the pointwise join and pointwise  order on $\ppset\Sta\to\ppset\Sta$. Please note the order is defined in terms of set inclusion at the outer layer of sets and is independent of the order on $\pset\Sta$.  
Define $\pbot = \dirimg{\bot}$ and note that 
$\pbot\emptyset = \emptyset$ and $\pbot\Q=\{\emptyset\}$ for $\Q\neq\emptyset$.
%The least element is the function $\lambda \Q \aftqua \{\emptyset\}$, abbreviated as $\pbot$.

\paragraph{Surprises.}

For semantics using h-transformers, some obvious guesses work fine but others do not.
The semantics in Fig.~\ref{fig:powtransem} 
uses operators $\ijoin$, $\hycond{b}$ and $\Dpset$ which will be explained in due course.  
For boolean expressions we simply lift by direct image, defining 
\( \meansH{b} = \dirimg{\meansT{b}} \).
The same for command atoms, so the semantics of $atm$ is derived from the given $\meansR{atm}$.

\begin{figure}[t]
\begin{normalsize}
\[
\begin{array}{lcl}
\meansH{atm}             &\; = \;& \dirimg{\meansT{atm}} \\
\meansH{\skipc}          & = & id \\ % id_{\ppset\Sta} \\
\meansH{c\, ; d}           & = & \meansH{c} \c \meansH{d} \\
\meansH{c \choice d}      & = & \meansH{c} \ijoin \meansH{d} \\ 
\meansH{ \ifc{b}{c}{d} } & = &  \meansH{c} \hycond{b} \meansH{d} \\
\meansH{ \whilec{b}{c} } & = & \lfp \H \\
               &\multicolumn{2}{l}{ 
         \mbox{where } \H: (\Dpset(\pset\Sta)\to\Dpset(\pset\Sta))\to (\Dpset(\pset\Sta)\to\Dpset(\pset\Sta))
                           \mbox{ is defined } } \\
               &\multicolumn{2}{l}{ \H \Phi = \meansH{c}\c\Phi \hycond{b} \meansH{\skipc} }
\end{array}
\]
\end{normalsize}
\caption{H-transformer semantics  $\protect\meansH{c} \in \Dpset(\pset\Sta)\to\Dpset(\pset\Sta)$.
}\label{fig:powtransem}
\end{figure}

The analog of Proposition~\ref{prop:reltran} is that 
for all $c$ in the signature, $\dirimg{\meansT{c}} = \meansH{c}$,
allowing laws valid in relational semantics to be lifted to h-transformers.
Considering some cases suggests that this could be proved by induction on $c$:

\begin{itemize}
\item $\skipc$: $\dirimg{\meansT{\skipc}} = \dirimg{id_{\pset\Sta}} = \meansH{\skipc} $
by definitions and using that $\dirimg{-}$ preserves identity.

\item $atm$: $\dirimg{\meansT{atm}} = \meansH{atm}$ by definition.

\item $c ; d$: $\dirimg{\meansT{c\c d}}
= \dirimg{\meansT{c}\c\meansT{d}}
= \dirimg{\meansT{c}}\c\dirimg{\meansT{d}}
= \meansH{c}\c\meansH{d}
= \meansH{c ; d}$
by definitions, distribution of $\dirimg{-}$ over $\c$, and putative induction hypothesis. 

\end{itemize}
These calculations suggest we may succeed with this obvious guess:
\begin{equation}\label{eq:badcond}
\meansH{ \ifc{b}{c}{d} }  =   \meansH{b} \c \meansH{c} \:\join\: \meansH{\neg b} \c \meansH{d} 
\end{equation}
The induction hypothesis would give 
\( \meansH{ \ifc{b}{c}{d} }  =   \dirimg{\meansT{b} \c \meansT{c}} \join \dirimg{\meansT{\neg b} \c \meansT{d}} \).
On the other hand,
\( \dirimg{\meansT{ \ifc{b}{c}{d} }} = \dirimg{\meansT{b} \c \meansT{c} \join \meansT{\neg b} \c \meansT{d}} \).
Unfortunately these are quite different because the joins are at different levels. 
In general, for $\phi$ and $\psi$ of type $\pset\Sta\to\pset\Sta$ and $\Q\in\ppset\Sta$ we have
$\dirimg{\phi \join \psi}\Q = \{ \phi\, p \union \psi\, p \mid p\in \Q \}$
whereas 
\( (\dirimg{\phi}\join\dirimg{\psi})\Q = \{ \phi\, p \mid p\in \Q \} \union \{ \psi\, p \mid p \in \Q \} \).   
Indeed, the same discrepancy would arise if we define
$\meansH{c \choice d} = \meansH{c} \join \meansH{d}$.

At this point one may investigate notions of ``inner join'', but for expository purposes we proceed to consider a putative definition for loops.
Following the pattern for relational and transformer semantics, an obvious guess is
\begin{equation}\label{eq:badloop}
\meansH{\whilec{b}{c}} = \lfp \K 
\mbox{ where }
\K \Phi = \meansH{b}\c\meansH{c}\c\Phi \join \meansH{\neg b}
\end{equation}
Consider this program: $\whilec{x<4}{x:=x+1}$. We can safely assume
$\meansH{x<4}$ is $\dirimg{\meansT{x<4}}$ and 
$\meansH{x:=x+1}$ is $\dirimg{\meansT{x:=x+1}}$.
As there is a single variable, we can represent a state by its value, for example $\{2,5\}$ is a set of two states.
Let us work out $\meansH{\whilec{x<4}{x:=x+1}} \{\{2,5\}\}$.
Now $\meansH{\whilec{x<4}{x:=x+1}}$ is the limit of the chain $\K^i\pbot$ where $\K^i$ means $i$ applications of $\K$.
Note that for any $\Phi$ and $i>0$, 
% Using disjunctivity to distribute 
\[ \K^i \Phi = 
   \begin{array}[t]{l}
       (\meansH{x<4}\c\meansH{x:=x+1})^i \c \Phi 
         \; \join \\
       (\join j :: 0\leq j < i \aftqua 
                  (\meansH{x<4}\c\meansH{x:=x+1})^j \c\meansH{\neg x<4} )
   \end{array}\]
Writing $\Q_i$ for $\K^i\pbot \{\{2,5\}\}$ one can derive 
\[\begin{array}{l}
\Q_0 = \{\emptyset\} \\
\Q_1 = \{\emptyset\} \union \{\{5\}\} = \{ \emptyset,\{5\}\} \\
\Q_2 = \{\emptyset\} \union \{\emptyset\}\union\{\{5\}\} = \{ \emptyset,\{5\}\} \\
\Q_3 = \{\emptyset\} \union \{\emptyset\}\union\{\{4\}\}\union\{\{5\}\} = \{ \emptyset,\{4\},\{5\}\}  
%\Q_4 = \{\emptyset\} \union \{\{4\}\}\union\{\emptyset\}\union\{\emptyset\}\union\{\{5\}\} = \{ \emptyset,\{4\},\{5\}\} 
\end{array}
\]
at which point the sequence remains fixed.
As in the case of conditional (\ref{eq:badcond}), the result is not consistent with the underlying semantics: 
\[ \meansT{\whilec{x<4}{x:=x+1}} \{2,5\} = \{4,5\} \]
The result should be $\{\{4,5\}\}$ if we are to have the analog of 
Proposition~\ref{prop:reltran}.

A plausible inner join is $\otimes$ defined
by $(\Phi\otimes\Psi)\Q = \{ r\union s \mid \some{q\in\Q}{r\in\Phi\{q\}\land s\in\Psi\{q\}} \}$.
This can be used to define a semantics of $\choice$
as well as semantics of conditional and loop;
the resulting constructs are $\rt$-monotonic and enjoy other nice properties.  

Indeed, using $\otimes$ in place of $\join$ in (\ref{eq:badloop}),
we get $\K^3\pbot\{\{2,5\}\} = \{\{4,5\}\}$, which is exactly the lift of the transformer semantics.
There is one serious problem: $\K$ fails to be increasing.
In particular, $\pbot\not\rt\K\pbot$; for example $\pbot\{\{2,5\}\} = \{\emptyset\}$ but
$\H\pbot\{\{2,5\}\} = \{\{5\}\}$.
While this semantics merits further study, we leave it aside because we aim to use fixpoint fusion results that rely on Kleene approximation: This requires $\pbot\rt\K\pbot$ in order to have an ascending chain, and the use of $\pbot$ so that $\dirimg{-}$ is strict. 

\paragraph{A viable solution.}

Replacing singleton by powerset in the definition of $\otimes$, for
any h-transformers $\Phi,\Psi : \ppset\Sta \to \ppset\Sta$ we define the inner join $\ijoin$ by 
\[ %begin{equation}\label{eq:ijoin}
(\Phi\ijoin\Psi)\Q
= \{ r\union s \mid \some{p\in\Q}{r\in\Phi(\pset{\,p})\land s\in\Psi(\pset{\,p})} \}
\]
For semantics of conditionals, it is convenient to define, for boolean expression $b$, 
this operator on h-transformers:
$\Phi\hycond{b}\Psi = \meansH{b}\c\Phi \ijoin \meansH{\neg b}\c\Psi$.
It satisfies 
\begin{equation}\label{eq:defhycond}
 (\Phi\hycond{b}\Psi)\Q = 
\{ r \union s \mid \some{p\in \Q}{
   r \in \Phi(\pset(\meansT{b}p)) \land 
   s \in \Psi(\pset(\meansT{\neg b}p)) } \} 
\end{equation}
because $\meansH{b}(\pset\,p) = \dirimg{\meansT{b}}(\pset\,p) = \pset(\meansT{b} p)$).
These operators are used in Fig.~\ref{fig:powtransem} for semantics of conditional and loop.

It is straightforward to prove $\ijoin$ is monotonic:
$\Phi\rt\Phi'$ and $\Psi\rt\Psi'$ imply $\Phi\ijoin\Psi \rt \Phi'\ijoin\Psi'$.
It is also straightforward to prove 
\begin{equation}\label{eq:joinijoin}
\dirimg{\phi\join\psi} \rt \dirimg{\phi}\ijoin\dirimg{\psi}
\end{equation}
but in general equality does not hold, so we focus on $\hycond{\args}$.

\begin{lemma}
For any $b$, $\hycond{b}$ is monotonic:
$\Phi\rt\Phi'$ and $\Psi\rt\Psi'$ imply
$\Phi\hycond{b}\Psi \rt \Phi'\hycond{b}\Psi'$.
\end{lemma}
\begin{proof}
Keep in mind this is $\rt$ at the outer level: 
$\Phi\rt\Phi'$ means $\all{\Q}{\Phi\Q \sbs \Phi'\Q}$ (more sets, not bigger sets, if you will).
This follows by monotonicity of $\ijoin$, or using characterization (\ref{eq:defhycond})
we have 
\[ r\union s \in (\Phi\hycond{b}\Psi)\Q
\quad\mbox{ iff }\quad 
\some{q\in\Q}{ r\in\Phi(\pset(\meansT{b}q)) \land 
                    s\in\Psi(\pset(\meansT{\neg b}q)) } 
\]
which implies 
$\some{q\in\Q}{ r\in\Phi'(\pset(\meansT{b}q)) \land 
                    s\in\Psi'(\pset(\meansT{\neg b}q)) }$
by $\Phi\rt\Phi'$ and $\Psi\rt\Psi'$.
\qed
\end{proof}

With $\meansH{ \ifc{b}{c}{d} }$ defined as in Fig.~\ref{fig:powtransem} we have the following refinement.

\begin{lemma}\label{lem:hcond}
$\dirimg{\meansT{ \ifc{b}{c}{d} }} \rt \meansH{ \ifc{b}{c}{d}}$
provided that 
$\dirimg{\meansT{c}} \rt \meansH{c}$ and $\dirimg{\meansT{d}} \rt \meansH{d}$.
\end{lemma}
\begin{proof}  

  {\  } \vspace*{-2ex}  % LATEX HACK ALERT

\begin{calcushort}
\flns{ \dirimg{\meansT{ \ifc{b}{c}{d} }} }
\plns{=}{ \dirimg{ \meansT{b}\c\meansT{c}\join\meansT{\neg b}\c\meansT{d} }}{semantics}
\plns{\rt}{ \dirimg{\meansT{b}\c\meansT{c}}\ijoin\dirimg{\meansT{\neg b}\c\meansT{d} }
    }{by (\ref{eq:joinijoin})}
\plns{=}{ \meansH{b}\c\dirimg{\meansT{c}}\ijoin\meansH{\neg b}\c\dirimg{\meansT{d} }
    }{distribute $\dirimg{-}$ over $\c$, semantics}         
\plns{\rt}{ \meansH{b}\c\meansH{c}\ijoin\meansH{\neg b}\c\meansH{d}}{assumption, monotonicity}
\plns{=}{ \meansH{ \ifc{b}{c}{d} }}{semantics, def of $\hycond{b}$ from $\ijoin$}
\end{calcushort}
\qed
\end{proof}

This result suggests that we might be able to prove $\dirimg{\meansT{c}}\rt\meansH{c}$ for all $c$, but that would be a weak link between the transformer and h-transformer semantics.  A stronger link can be forged as follows.

We say $\Q\in\ppset\Sta$ is \dt{subset closed} iff $\Q = \ssc\Q$
where the subset closure operator $\ssc$ is defined by $p\in\ssc\Q$ iff $\some{q\in\Q}{p\subseteq q}$.
For example, the set $\A$ used in (\ref{eq:NI}) is subset closed.
Observe that $\ssc = \dirimg{\rsps}$.

\begin{lemma}\label{lem:imghycond}
For transformers $\phi,\psi:\pset\Sta\to\pset\Sta$ and condition $b$,
if $\Q=\ssc \Q$ then
\( \dirimg{ \meansT{b}\c \phi \join \meansT{\neg b}\c \psi }\Q \;=\; (\dirimg{\phi} \hycond{b} \dirimg{\psi})\Q \).
\end{lemma}
\begin{proof}
For the LHS, by definitions:
\begin{calcushort}
\flns{    \dirimg{ \meansT{b}\c \phi \join \meansT{\neg b}\c \psi }\Q }
\plns{=}{ \{ r\union s \mid \some{q\in\Q}{ r = \phi(\meansT{b}q) \land s = \psi(\meansT{\neg b}q) } \}
    \hspace*{3em} (*)}{} %defs $\dirimg{-},\join$}
\end{calcushort}
For the RHS, again by definitions:
\begin{calcushort}
\flns{     (\dirimg{\phi} \hycond{b} \dirimg{\psi})\Q }
\plns{=}{ \{ r\union s \mid \some{q\in\Q}{ r \in \dirimg{\phi}(\pset(\meansT{b}q)) \land 
                                          s \in \dirimg{\psi}(\pset(\meansT{\neg b}q)) } \} }{}
\plns{=}{ \{ r\union s \mid \some{q\in\Q}{ \some{t,u}{
                                             t\sbs \meansT{b}q \land 
                                             u\sbs \meansT{\neg b}q \land 
                                             r = \phi t \land 
                                             s = \psi u }} \} 
    \hspace*{3em} (\dagger) }{}
\end{calcushort}
Now $(*)\sbs (\dagger)$ by instantiating $t := \meansT{b}q$ and $u := \meansT{\neg b}q$, so LHS $\sbs$ RHS is proved---as expected, given (\ref{eq:joinijoin}).
If $\Q$ is subset closed, we get $(\dagger)\sbs(*)$ as follows.  
Given $q,t,u$ in $(\dagger)$, let $q' := t\union u$.  
Then $t=\meansT{b}q'$ and  $u=\meansT{\neg b}q$ because $\meansT{b}$ and $\meansT{\neg b}$ are filters.
And $q'\in\Q$ by subset closure.
Taking $q:=q'$ in $(*)$ completes the proof of RHS $\sbs$ LHS.
\qed
\end{proof}

\paragraph{Preservation of subset closure.}

In light of Lemma~\ref{lem:imghycond}, we aim to restrict attention to h-transformers on subset closed sets.
To this end we introduce a few notations.
%Recall that $\Q\in\ppset A$ iff $\Q\subseteq \pset A$.
The subset-closed powerset operator $\Dpset$ is defined on powersets $\pset A$, 
by
\begin{equation}\label{eq:Dpset}
\Q\in\Dpset(\pset A) \quad\mbox{iff}\quad \Q\subseteq \pset A \mbox{ and } \Q=\ssc\Q \mbox{ and } \Q\neq\emptyset
\end{equation}
To restrict attention to h-transformers of type $\Dpset(\pset\Sta)\to\Dpset(\pset\Sta)$ we must show that subset closure is preserved by the semantic constructs.  

For any transformer $\phi$,
define $\PSC \phi$ iff $(\rsps\c \phi) = (\phi\c\rsps)$.  
The acronym is explained by the lemma to follow.
By definitions, the inclusion $(\rsps\c \phi) \sps (\phi\c\rsps)$ is equivalent to 
\begin{equation}\label{eq:psc}
\all{q,r}{ \phi\, q \sps r \imp \some{s}{q\sps s \land \phi\, s = r}}
\end{equation}
Recall from Section~\ref{sec:rel} that the reverse,
$(\rsps\c \phi) \sbs (\phi\c\rsps)$, is monotonicity of $\phi$.

\begin{lemma}\label{lem:psc}
$\PSC\phi$ implies $\dirimg{\phi}$ preserves subset closure.
\end{lemma}
\begin{proof}
For any subset closed $\Q$, $\dirimg{\phi}\Q$ is subset closed because
$\dirimg{\rsps}(\dirimg{\phi}\Q)    
= \dirimg{\phi\c\rsps} \Q 
= \dirimg{\rsps\c\phi} \Q 
= \dirimg{\phi}(\dirimg{\rsps}\Q)
= \dirimg{\phi}\Q
$
using functoriality of $\dirimg{-}$, $\PSC \phi$, and $\ssc\Q = \Q$.
\qed
\end{proof}
It is straightforward to show $\PSC\bot$.
The following is a key fact, but also a disappointment that leads us
away from nondeterminacy.

\begin{lemma}\label{lem:pscR}
If $R$ is a partial function (i.e., $\converse{R}\c R\sbs id$)
then $\PSC\dirimg{R}$.
\end{lemma}
\begin{proof}
In accord with (\ref{eq:psc}) we show for any $q,r$ that 
$ r \sbs \dirimg{R} q \imp \some{s\sbs q}{\phi s = r}$.
Suppose $r \sbs \dirimg{R} q$. 
Let $s=(\dirimg{\converse{R}} r)\intersect q$, so for any $x$ we have
$x\in s$ iff $x\in q$ and $\some{y\in r}{x R y}$.
We have $s\sbs q$ and it remains to show $\dirimg{R} s = r$, which holds because for any $y$
\begin{calcushort}
\flns{ y\in\dirimg{R} s }
\plns{\equiv}{ \some{x}{x \in s \land xRy} }{def $\dirimg{-}$}
\plns{\equiv}{ \some{x}{x\in q \land (\some{z}{z\in r \land xRz}) \land xRy }}{def $s$}
\plns{\equiv}{ \some{x}{x\in q \land y \in r \land xRy} }{$R$ partial function}
\plns{\equiv}{ y \in r}{$\impby$ by $r\sbs \dirimg{R}q$ and def $\dirimg{-}$}
\end{calcushort}
%\qed
\end{proof}

Using dots to show domain and range elements, the diagram on the left 
is an example $R$ such that $\PSC\dirimg{R}$ but $R$ is not a partial function.
The diagram on the right is a relation, the image of which does not satisfy $\PSC$.
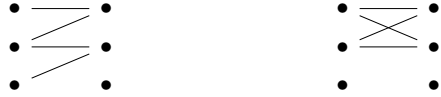
\begin{footnotesize}
\begin{center}
      \begin{tikzcd}
         \point \arrow[r,"",dash] %\arrow[r,"",dashed,dash,shift left] 
& \point\\
         \point \arrow[r,"",dash] %\arrow[r,"",dashed,dash,shift left]
\arrow[ru,"",dash] 
& \point\\
         \point \arrow[ru,"",dash] %\arrow[ru,"",dashed,dash,shift left] 
& \point
      \end{tikzcd}
\qquad\qquad\qquad\qquad
      \begin{tikzcd}
         \point \arrow[r,"",dash] \arrow[rd,"",dash] & \point\\
         \point \arrow[r,"",dash] \arrow[ru,"",dash] & \point\\
         \point & \point 
      \end{tikzcd}
\end{center}
\end{footnotesize}

As a consequence of Lemmas~\ref{lem:psc} and~\ref{lem:pscR} we have the following.
% from noteSSC
\begin{lemma}\label{lem:imgssc}
If $R$ is a partial function then $\dirimg{\dirimg{R}}:\ppset{A}\to\ppset{B}$ preserves subset closure.
\end{lemma}

\paragraph{The theorem.}

To prove $\meansH{c} = \dirimg{\meansT{c}}$,
we want to identify a subset of $\pset\Sta\to\pset\Sta$ 
satisfying two criteria.
First, $\meansT{\args}$ can be defined within it, so in particular it is closed under $\G$ in Figure~\ref{fig:transem}.
Second, on the subset, $\dirimg{-}$ is strict and continuous 
into $\Dpset(\pset\Sta)\to\Dpset(\pset\Sta)$, to enable the use of fixpoint fusion.
Strictness is the reason\footnote{In~\cite{ClarksonS10}, 
other reasons are given for using 
$\{\emptyset\}$ rather than $\emptyset$ as the false hyperproperty.}
to disallow the empty set in (\ref{eq:Dpset});
it makes $\pbot$ (which equals $\dirimg{\bot}$) the least element, whereas otherwise the least element 
would be $\lambda \Q \aftqua \emptyset$.
We need the subset to be closed under pointwise union, at least for chains, so that $\dirimg{-}$ is continuous. 

Given that $\dirimg{R}$ is universally disjunctive for any $R$,
Proposition~\ref{prop:reltran} suggests restricting to universally disjunctive transformers.
Lemma~\ref{lem:psc} suggests restricting to transformers satisfying $\PSC$.
But we were not able to show the universally disjunctive transformers satisfying $\PSC$ are closed under limits.
We proceed as follows.

Define $\Dom\, \phi = \{ x \mid \phi\{x\} \neq \emptyset \}$ and note that $\Dom\dirimg{R}=\rdom R$ where $\rdom R$ is the usual domain of a relation.  
By a straightforward proof we have:

\begin{lemma}\label{lem:domInter}
For universally disjunctive $\phi$ and any $r$ we have $\phi\, r = \phi(r\intersect \Dom\, \phi)$. 
\end{lemma}

% (Also have, for any $\emptyset$-strict transformer $\phi$, $\Dom(\phi\c\psi) = \Dom\, \phi$ but not using that.)

\begin{lemma}\label{lem:domDisj}
For universally disjunctive $\phi,\psi$ with $\Dom\,\phi\intersect\Dom\,\psi = \emptyset$,
if $\PSC \phi$ and $\PSC \psi$ then $\PSC(\phi\join\psi)$.
\end{lemma}
\begin{proof}
For any $q,r$ with $r \sbs (\phi\join\psi)q$ we need to show 
$\some{s\sbs q}{(\phi\join\psi)s = r}$.
First observe 
\begin{calcu}
\fln{ r \sbs (\phi\join\psi)q }
\pln{\equiv}{ r \sbs \phi\, q \union \psi\, q }{def $\join$}
\pln{\equiv}{ r \sbs \phi (q\intersect \Dom\,\phi) \union \psi(q\intersect \Dom\,\psi)}{Lemma \ref{lem:domInter}}
\pln{\imp}{ r = s \union s' \land s \sbs \phi(q\intersect \Dom\,\phi) 
                            \land s' \sbs \psi(q\intersect \Dom\,\psi)}{set theory, 
   letting  $s = r\intersect \phi (q\intersect \Dom\,\phi)$
   and $s' = r\intersect \psi(q\intersect \Dom\,\psi)$}
\pln{\imp}{ \some{t,t'}{t\sbs q\intersect\Dom\,\phi \land t'\sbs q\intersect\Dom\,\psi \land
               \phi\, t = s \land \psi\, t' = s' \land s\union s' = r } \quad (*)}{using $\PSC\phi$ and $\PSC\psi$}
\end{calcu}
We use $(*)$ to show that $t\union t'$ witnesses $\PSC(\phi\join\psi)$, as follows:
$ (\phi\join\psi)(t\union t') 
= \phi(t\union t') \union \psi(t\union t') 
= \phi t \union \psi t'
= s \union s' 
= r
$
using also the definition of $\join$, and 
$\phi(t\union t')=\phi\,t$ and 
$\psi(t\union t')=\psi\,t'$ from Lemma~\ref{lem:domInter} and $(*)$.
\qed
\end{proof}

\begin{lemma}\label{lem:hycondSSC}
If $\Phi$ and $\Psi$ preserve subset closure then 
$(\Phi\hycond{b}\Psi)\Q$ is subset closed (regardless of whether $\Q$ is).
\end{lemma}
\begin{proof}
Suppose 
$q$ is in $(\Phi\hycond{b}\Psi)\Q$ and 
$q' \sbs q$.  
So according to (\ref{eq:defhycond})
there are $r,s,p$ with 
$p\in \Q$,
$q=r\union s$,
$r \in \Phi(\pset(\meansT{b}p))$, and 
$s \in \Psi(\pset(\meansT{\neg b}p))$.
Let $r'=r\intersect q'$ and 
    $s'=s\intersect q'$, so $r'\sbs r$ and $s'\sbs s$.
Because powersets are subset closed,
$\Phi(\pset(\meansT{b}p))$ and 
$\Psi(\pset(\meansT{\neg b}p))$ are subset closed, hence 
$r'\in\Phi(\pset(\meansT{b}p))$ and 
$s'\in\Psi(\pset(\meansT{\neg b}p))$.
As $q'=r'\union s'$, 
we have $q' \in (\Phi\hycond{b}\Psi)\Q$.
\qed
\end{proof}

It is straightforward to prove that $\Phi\ijoin\Psi$ preserves subset closure if $\Phi,\Psi$ do, similar to the proof of Lemma~\ref{lem:hycondSSC}.
By contrast, $\Phi\otimes\Psi$ does not preserve subset closure even if $\Phi$ and $\Psi$ do.
%\dn{As an aside, the inequality (\ref{eq:joinijoin}) can be shown to be strict because even if universally disjunctive $\phi$ and $\psi$ preserve subset closure, 
%their pointwise join need not.  (Compare Lemma~\ref{lem:domDisj}). }

Next, we confirm that $\meansH{\args}$ can be defined within the monotonic functions 
$\Dpset(\pset\Sta)\to\Dpset(\pset\Sta)$.

\begin{lemma}\label{lem:meansHpsc}
For all $c,\Q$, if $\Q$ is subset closed then so is $\meansH{c}\Q$, provided that
$\PSC\dirimg{\meansR{atm}}$ for every $atm$.
\end{lemma}
\begin{proof}
By induction on $c$.
\begin{itemize}
\item 
$atm$: 
$\meansH{atm}$ is $\dirimg{\dirimg{\meansR{atm}}}$ so by assumption
$\PSC\dirimg{\meansR{atm}}$ and Lemma~\ref{lem:psc}.

\item 
$\skipc$: immediate.

\item 
$c ; d$: by definitions and induction hypothesis.

\item 
  $c\choice d$: by induction hypothesis and observation above about $\ijoin$.

\item 
$\ifc{b}{c}{d}$: by Lemma~\ref{lem:hycondSSC} and induction hypothesis.

\item 
$\whilec{b}{c}$:  %by fixpoint induction using Lemma~\ref{lem:hycondSSC} and induction hypothesis.
Because $\pbot$ is least in $\Dpset(\pset\Sta)\to\Dpset(\pset\Sta)$, we have $\pbot\rt\H\pbot$,
so using monotonicity of $\H$ we have Kleene iterates.
Suppose $\Q$ is subset closed. 
To show $\lfp \H \, \Q$ is subset closed, note that $\lfp \H = \H^\gamma \, \Q$ where $\gamma$ is some ordinal. We show that $\H^\alpha \, \Q$ is subset closed, for every $\alpha$ up to $\gamma$,
by ordinal induction.
\begin{itemize}
\item $\H^0 \, \Q = \Q$ which is subset closed.

\item $\H^{\alpha+1} \, \Q = (\meansH{c}\c \H^\alpha \hycond{b} \meansH{\skipc})\Q$ by definition of $\H$.  
Now $\H^\alpha$ preserves subset closure by the ordinal induction hypothesis,
and $\meansH{c}$ preserves subset closure by the main induction hypothesis.
So $\meansH{c}\c \H^\alpha$ preserves subset closure, as does $\meansH{\skipc}$.
Hence $\meansH{c}\c \H^\alpha \hycond{b} \meansH{\skipc}$ preserves subset closure 
by Lemma~\ref{lem:hycondSSC}.

\item $\H^\beta \, \Q = (\join_{\alpha<\beta} H^\alpha)\Q$ (for non-0 limit ordinal $\beta$),
which in turn equals 
$\union_{\alpha<\beta} (H^\alpha \, \Q)$ because $\join$ is pointwise.  
By induction, each $H^\alpha \, \Q$ is subset closed, and closure is preserved by union, 
so we are done.
\qed
\end{itemize}
\end{itemize}
\end{proof}

Returning to the two criteria for a subset of $\pset\Sta\to\pset\Sta$, 
suppose $\meansR{atm}$ is a partial function, for all $atm$ --- in short,
\dt{atoms are deterministic}. 
If in addition $c$ is $\choice$-free, then $\meansR{c}$ is a partial function.  
Under these conditions, by Proposition~\ref{prop:reltran}, $\meansT{c}$ is the direct image of a partial function.  

Let $\IPF$ be the subset of $\pset\Sta\to\pset\Sta$ that are direct images of partial functions,
i.e.,
$\IPF = \{ \phi\in \pset\Sta\to\pset\Sta \mid
         \some{R}{\phi=\dirimg{R} \mbox{ and } \converse{R}\c R\subseteq id }  \}$.
Observe that $\IPF$ is closed under $\G$,
because for $\phi\in \IPF$ with $\phi=\dirimg{R}$ we have $\G\dirimg{R} 
= \dirimg{\meansR{b}}\c\dirimg{\meansR{c}}\c\dirimg{R} \join
                        \dirimg{\meansR{\neg b}}
= \dirimg{\meansR{b}\c\meansR{c}\c R} \join
                        \dirimg{\meansR{\neg b}}
= \dirimg{\meansR{b}\c\meansR{c}\c R \union \meansR{\neg b}}$
and the union is of partial functions with disjoint domains
so it is a partial function.
We have $\bot\rt \G\bot$ because $\bot$ is the least element
in $\pset\Sta\to\pset\Sta$.
%and then by monotonicity of $\G$ we get an increasing chain.
By Lemma~\ref{lem:imgssc}, when $\dirimg{-}$ is restricted to $\IPF$, its range is included in 
$\Dpset(\pset\Sta)\to\Dpset(\pset\Sta)$.
In $\IPF$, lubs of chains are given by pointwise union, 
so $\dirimg{-}$ is a strict and continuous function 
from $\IPF$ to $\Dpset(\pset\Sta)\to\Dpset(\pset\Sta)$.

To state the theorem,
we write $\dot{=}$ for extensional equality 
on h-transformers of type  $\Dpset(\pset\Sta)\to\Dpset(\pset\Sta)$,
i.e., equal results on all subset closed $\Q$.

\begin{theorem}\label{thm:main}
%For all $\choice$-free $c$, $\dirimg{\meansT{c}} \; \dot{=} \; \meansH{c}$, provided atoms are deterministic.
$\dirimg{\meansT{c}} \; \dot{=} \; \meansH{c}$, provided atoms are deterministic and $c$ is $\choice$-free.
\end{theorem}
\begin{proof}
By induction on $c$.  
For the cases of $\skipc$, atoms, and $\c$ the arguments preceding (\ref{eq:badcond}) are still valid.
For conditional, observe
\begin{calcushort}
\flns{ \dirimg{ \meansT{\ifc{b}{c}{d}}}}
\plns{\dot{=}}{ \dirimg{ \meansT{b}\c\meansT{c} \join \meansT{\neg b}\c\meansT{d} }}{semantics}
\plns{\dot{=}}{ \dirimg{\meansT{c}} \hycond{b} \dirimg{\meansT{d}} }{Lemma~\ref{lem:imghycond}} 
\plns{\dot{=}}{ \meansH{c} \hycond{b} \meansH{d} }{induction hypothesis}
\plns{\dot{=}}{ \meansH{\ifc{b}{c}{d}} }{semantics}
\end{calcushort}
Finally, the loop:
\begin{calcushort}
\flns{ \dirimg{ \meansT{\whilec{b}{c}} }}
\plns{\dot{=}}{ \dirimg{ \lfp \G }}{semantics}
\plns{\dot{=}}{ \lfp \H }{fixpoint fusion, see below}
\plns{\dot{=}}{ \meansH{\whilec{b}{c}} }{semantics}
\end{calcushort}
The antecedent for fusion is 
\( \all{\phi}{ \dirimg{\G \phi} = \H\dirimg{\phi} } \)
and it holds because for any $\phi$:
\begin{calcushort}
\flns{ \dirimg{ \G \phi } }
\plns{\dot{=}}{ \dirimg{ \meansT{b}\c\meansT{c}\c\phi \join \meansT{\neg b} }}{def $\G$}
\plns{\dot{=}}{ \dirimg{ \meansT{b}\c\meansT{c}\c\phi \join \meansT{\neg b}\c\meansT{\skipc} }}{skip law}
\plns{\dot{=}}{ \dirimg{ \meansT{c}\c\phi } \hycond{b} \dirimg{ \meansT{\skipc} } }{Lemma~\ref{lem:imghycond}}
\plns{\dot{=}}{ \dirimg{ \meansT{c} }\c\dirimg{ \phi } \hycond{b} \dirimg{\meansT{\skipc}} }{$\dirimg{-}$ distributes over  $\c$}
\plns{\dot{=}}{ \meansH{c}\c\dirimg{ \phi } \hycond{b} \meansH{\skipc} }{induction hypothesis}
\plns{\dot{=}}{ \H\dirimg{\phi} }{def $\H$}
\end{calcushort}
%\qed
\end{proof}

\section{Specifications and refinement}\label{sec:spec}

We wish to conceive of specifications as miraculous programs that can 
achieve by refusing to do, can choose the best angelically, and can compute the uncomputable.  
We wish to establish rigorous connections between programs and specifications, 
perhaps by deriving a program that can be automatically compiled for execution,
perhaps by deriving a specification that can be inspected to determine the usefulness or trustworthiness of the program.
A good theory may enable automatic derivation in one direction or the other, 
but should also account for ad hoc construction of proofs.
Simple reasons should be expressed simply, so algebraic laws and transitive refinement chains are important.  
In this inconclusive section, we return to the general notion of hyperproperty and consider how the h-transformer semantics sheds light on refinement for hyperproperties.  Initially we leave aside the signature/semantics notations.

Let $R:\Sta\relto\Sta$ be considered as a program, and $\Hp$ be a hyperproperty, that is, $\Hp$ is a set of programs.  Formally: $\Hp\in\pset(\Sta\relto\Sta)$.  For $R$ to satisfy $\Hp$ means, by definition, that $R\in\Hp$.
The example of possibilistic noninterference shows that in general trace refinement is unsound: $R\in \Hp$ does not follow from 
$S\in \Hp$ and $S\sps R$.
It does follow in the case that $\Hp$ is subset closed.  Given that $\pset(\Sta\relto\Sta)$ is a huge space, one may hope that specifications of practice interest may lie in relatively tame subsets.  Let us focus on subset closed hyperproperties,
for which one form of chain looks like
\begin{equation}\label{eq:firstChain}
\Hp \;\Ni\;  S \sps \ldots \sps T \sps R 
\end{equation}
Although this is a sound way to prove $R\in\Hp$, it does not seem sufficient, at least for examples like $\NI$ which require some degree of determinacy.  The problem is that for
intermediate steps of trace refinement it is helpful to use nondeterminacy for the sake of abstraction and underspecification, so finding suitable $S$ and $T$ may be difficult. 
One approach to this problem is to use a more nuanced notion of refinement, that preserves a hyperproperty of interest.
For confidentiality, Banks and Jacob explore this approach in the setting of UTP~\cite{BanksJ10a}.

Another form of chain looks like 
\begin{equation}\label{eq:secondChain}
 \Hp \sps \mathbb{S} \sps \ldots \sps \mathbb{T} \;\Ni\; R 
\end{equation}
where most intermediate terms are at the hyper level, i.e., $\mathbb{S}$ and $\mathbb{T}$ are, like $\Hp$, elements of 
$\pset(\Sta\relto\Sta)$. 
The chain is a sound way to prove $R\in \Hp$, even if $\Hp$ is not subset closed.  
But in what way are the intermediates $\mathbb{S},\mathbb{T},\ldots$ expressed, and by what reasoning are the containments established?  What is the relevant algebra, beyond elementary set theory?

The development in Section~\ref{sec:htrans} is meant to suggest a third form of chain:
\begin{equation}\label{eq:thirdChain}
  \ldots \rf \mathbb{S} \rf \ldots \rf \mathbb{T} \rf \dirimg{\dirimg{R}} 
\end{equation}
Here the intermediate terms are of type $\Dpset(\pset\Sta)\to\Dpset(\pset\Sta)$ and $\rt$ is the pointwise ordering (used already in Section~\ref{sec:htrans}).   
The good news is that if the intermediate terms are expressed using program notations, they may be amenable to familiar laws such as those of Kleene algebra with 
tests~\cite{Kozen00,Struth16}, for which relations are a standard model.  
A corollary of Theorem~\ref{thm:main} is that the laws hold for deterministic terms expressed in the signature (\ref{eq:signature}).  
To make this claim precise one might spell out the healthiness conditions of elements in the range of $\meansH{\args}$,  
but more interesting would be to extend the language with specification constructs, using (in)equational conditions like 
$\converse{R}\c R \sbs id$ as antecedents in conditional laws of healthy fragments.  
We leave this aside in order to focus on a gap in our story so far.

%FUTURE spell out that it's a model of KAT

The third form of chain is displayed with elipses on the left because we lack an account of specifications! 
Our leading example, $\NI$, is defined as a set of relations, whereas 
in (\ref{eq:thirdChain})
the displayed chain needs the specification, say $\Psi$, to have type $\Dpset(\pset\Sta)\to\Dpset(\pset\Sta)$.  
The closest we have come is the characterization  $R\in\NI$ iff $\dirimg{\dirimg{R}}\A \sbs \A$, see (\ref{eq:NI}).
But this is a set containment, whereas we seek $\Psi$ with
$R\in\NI$ iff $\Psi\rf\dirimg{\dirimg{R}}$.
In the rest of this section we sketch two ways to proceed.

On the face of it, $\Psi\rf\dirimg{\dirimg{R}}$ seems problematic because $\rf$ is an ordering on functions. 
Given a particular set $p\in\pset\Sta$ with $\Agrl p$, a noninterfering  $R$ makes specific choice of value for $lo$ whereas specification $\Psi$ should allow any value for $lo$ provided that the choice does not depend on the initial value for $hi$.  
One possibility is to escalate further and allow the specification to be a relation $\Dpset(\pset\Sta)\relto\Dpset(\pset\Sta)$.  
To define such a relation, first lift the predicate $\Agrl$ on sets to the filter 
$\widehat{\Agrl}: \pset(\pset\Sta)\relto\pset(\pset\Sta)$ defined by
\[ \widehat{\Agrl}\,\Q = \{ p\in \Q \mid \Agrl\, p \} \]
Note that  $\A = \widehat{\Agrl}(\pset\Sta)$.  
More to the point, $\widehat{\Agrl}\,\Q = \Q$ just if each $p\in\Q$ satisfies $\Agrl$.
Now define $\mathbf{NI}$,
as a relation 
$\mathbf{NI} : \Dpset(\pset\Sta)\relto\Dpset(\pset\Sta)$,  
by 
\[ \P \:\mathbf{NI}\: \Q \quad\mbox{iff}\quad  \widehat{\Agrl}\,\P = \P \imp \widehat{\Agrl}\,\Q = \Q \]
This achieves the following: $R\in\NI$ iff $\mathbf{NI}\sps \dirimg{\dirimg{R}}$.
But this inclusion does not compose transitively with $\rf$ in the third form of chain,
so we proceed no further in this direction.

The second way to proceed can be described using a variation on the h-transformer semantics
of Section~\ref{sec:htrans}.
It will lead us back to 
the second form of chain, (\ref{eq:secondChain}), in particular for $\NI$ as $\Hp$.
The idea resembles UTP models of reactive processes~\cite[Chapt.~8]{HoareHe},
in which an event history is related to its possible extension.
Here we use just pre-post traces, as follows.
Let $\Trc = \Sta\times\Sta$.  
Consider a semantics $\meansT{\args}'$ such that 
$\meansT{c}'$ has type $\pset\Trc\to\pset\Trc$. 
Instead of transforming an initial state to a final one
(or rather, state set as in $\meansT{\args}$),
an initial trace $(\sigma,\tau)$ is mapped to traces $(\sigma,\upsilon)$ for $\upsilon$ with $\tau \meansR{c} \upsilon$.  
The semantics $\meansT{\args}'$ is not difficult to define (or see~\cite[sec.~2]{AssafNSTT17}).
The upshot is that for $S\in\pset\Trc$,
the trace set $\meansT{c}' S$ is the relation $S\c\meansR{c}$.
In particular, let $init$ be $id_{\Sta}$, viewed as an element of $\pset\Trc$.
We get $\meansT{c}'init = \meansR{c}$, the relation denoted by $c$.
Lifting, we obtain a semantics $\meansH{\args}'$,
at the level $\ppset\Trc\to\ppset\Trc$,
such that 
$\meansH{c}'\{init\}$ contains $\meansR{c}$.  
(In light of Theorem~\ref{thm:main} and the discussion preceding it, we do not expect
$\meansH{c}'\{init\}$ to be just the singleton $\{ \meansR{c} \}$, as $\{init\}$ is not subset closed.)
This suggests a chain of the form
$\NI \sps \ldots \sps \Psi(\ssc\{init\}) \sps \meansH{c}'(\ssc\{init\}) \,\Ni\, \meansR{c} $
which proves that $\meansR{c}$ satisfies $\NI$ --- and which may be derived from a subsidiary chain of refinements 
like $\Psi \rf \meansH{c}'$ as in the third form of chain, independent of the argument $\ssc\{init\}$.

This approach has been explored in the setting of abstract interpretation, where the intermediate terms are obtained as a computable approximation of a given program's semantics.  To sketch the the idea we first review abstract interpretation for trace properties.
Mathematically, abstract interpretation is very close to data refinement, where intermediate steps involve changes of data representation.
For example, the state space $\Sta$ of $R:\Sta\relto\Sta$ would be connected with another, say $\Delta$, 
by a coupling relation $\rho:\Delta\relto\Sta$ subject to a simulation condition such as 
$S\c\rho \sps \rho\c R$, recall (\ref{eq:sim}).
With a functional coupling, the connection could be $\rho(\dirimg{S}\Delta) \rf \dirimg{R}(\rho\Delta)$.  

Let $T\in \Sta\relto\Sta$ be a trace set intended as a trace property specification.
In terms of the trace-computing semantics $\meansT{\args}'$ above, $c$ satisfies $T$ provided that
$T \sps \meansT{c}' init$.
It can be proved by the following chain, ingredients of which are to be explained.
\[
T \;\sps\; \gamma(\meansT{c}^\sharp \, a) 
  \;\sps\; \meansT{c}'(\gamma\, a) 
  \;\sps\; \meansT{c}' init 
\]
Here $\gamma:A\to\pset\Trc$ is like $\rho$ above, mapping some convenient domain $A$ to traces.
The element $a\in A$ is supposed to be an approximation of the initial traces, i.e.,
$\gamma\,a \sps init$.
Thus the containment $\meansT{c}'(\gamma\,a) \sps \meansT{c}' init$ is by monotonicity of 
semantics.  
The next containment, $\gamma(\meansT{c}^\sharp a) \sps \meansT{c}'(\gamma a)$, involves 
an ``abstract'' semantics $\meansT{c}^\sharp : A \to A$.
Indeed, the containment is the soundness requirement for such semantics. 
What remains is the containment
$T \sps \gamma(\meansT{c}^\sharp \, a)$ which needs to be checked somehow.  
Typically, $\gamma$ is part of a Galois connection, i.e., it has a lower adjoint $\alpha:\pset\Trc\to A$
and the latter check is equivalent to 
$\alpha T \geq \meansT{c}^\sharp \, a$ where $\geq$ is the order on $A$.
Ideally it is amenable to automation, but that is beside the point.

The point is to escalate this story to the hyper level, in a chain of this form:
\[ %begin{equation}\label{eq:aichain}
\Hp \;\sps\; \gamma(\meansH{c}^\sharp \, a) 
  \;\sps\; \meansH{c}'(\gamma\, a) 
  \;\sps\; \meansH{c}'(\ssc\{init\})
  \;\,\Ni\;\,  \meansT{c}' init
\] 
Now $\gamma$ has type $A\to\ppset\Trc$.
Again, the abstract semantics should be sound --- condition $\gamma(\meansH{c}^\sharp \, a) \sps \meansH{c}'(\gamma\, a)$ --- now with respect to a set-of-trace-set semantics $\meansH{\args}'$ that corresponds to our Fig.~\ref{fig:powtransem}.
The element $a\in A$ now approximates the set $\ssc\{init\}$
and $\meansH{c}'(\gamma\, a) \sps \meansH{c}'(\ssc\{init\})$ is by monotonicity.
The step $\Hp \sps \gamma(\meansH{c}^\sharp \, a)$ may again be checked at the abstract level
as $\alpha \Hp \geq \meansH{c}^\sharp \, a$.
Of course $\Hp$ is a hyperproperty so the goal is to prove the program is an element of $\Hp$.
This follows provided that 
$\meansH{c}'(\ssc\{init\}) \; \Ni \;  \meansT{c}' init$,
a connection like our Theorem~\ref{thm:main} except for moving from $\pset\Sta$ to $\pset\Trc$.

\section{Related work}\label{sec:related}

The use of algebra in unifying theories of programming has been explored in many works including the book that led to the UTP meetings~\cite{HoareHe,Sampaio97,HoareMSW11}.
Methodologically oriented works include the books by Morgan~\cite{Morgan:book} and by Bird and de 
Moor~\cite{BirdDeM}.

The term hyperproperty was introduced by Clarkson and Schneider who among other things mention that refinement at the level of trace properties is admissible for proving subset closed hyperproperties~\cite{ClarksonS10}.
They point out that the topological classification of trace properties, i.e., safety and liveness, corresponds to similar notions dubbed hypersafety and hyperliveness.  Subset closed hyperproperties strictly subsume hypersafety.
As it happens, $\NI$ is in the class called 2-safety that specifies a property as holding for every pair of traces.    
For fixed $k$, one can encode $k$-safety by a product program, each trace of which represents $k$ traces of the original.  
Some interesting requirements, such as quantitative information flow, are not in $k$-safety for any $k$.

Epistemic logic is the topic of a textbook~\cite{HalpernEtalKnowl95}
and has been explored in the security literature~\cite{BalliuDG11}.
Mantel considers a range of security properties via closure operators~\cite{Mantel02}.
The limited usefulness of trace refinement for proving $\NI$ even for deterministic programs, as in the chain (\ref{eq:firstChain}), is discussed by Assaf and Pasqua~\cite{AssafNSTT17,MastroeniP18}.
The formulation of possibilistic noninterference as 
$\lagree \c R \c \lagree \; = \; R\c\lagree$ is due to Joshi and Leino~\cite{Joshi:Leino:SCP00}
and resembles the formulation of Roscoe et al.~\cite{RoscoeWW1994}.
% They cite Roscoe 1995, though withouc clearly saying their idea is essentially
% the same as RoscoeWW1994.

The textbook of Back and von Wright~\cite{Back:book} explores predicate transformer semantics and refinement calculus.  
The use of $sglt$ and $\Ni$ in (\ref{eq:relRecover}) is part of the extensive algebra connecting predicate transformers and relations using categorical notions~\cite{MartinSCP}.
Upward closed sets of predicates play an important role in that algebra~\cite{Naumann98},
which should be explored in connection with the present investigation and its potential application to higher order programs.
The extension of functional programming calculus to imperative refinement is one setting in which strong laws (Cartesian closure) for a well-behaved subset are expressed 
as implications with inequational antecedents~\cite{Naumann95a,Naumann15}.
These works use backward predicate transformers in order to model general specifications and in particular the combination of angelic and demonic nondeterminacy.
Alternative models with similar aims can be found in
Martin \emph{et al}~\cite{MartinCR07} and Morris \emph{et al}~\cite{MorrisBT09}.

A primary precursor to this paper is the dissertion work of Assaf, which targets refinement chains in the style of abstract interpretation~\cite{CousotCousot79,Cou99}.
Assaf's work~\cite{AssafNSTT17} 
introduced a set-of-sets lifted semantics 
from which our h-transformer semantics is adapted.
In keeping with the focus on static analysis, Assaf shows the lifted semantics is an approximation of the underlying one.\footnote{Assaf \emph{et al} use fixpoint fusion in the inequational form mentioned following (\ref{eq:fusion}), 
to prove soundness of the derived abstract semantics.
Their inequational result corresponding to our Theorem is proved, in the loop case, using explicit induction on approximation chains.  
See the proof of Theorem~1 in~\cite{AssafNSTT17}.}
%or rather the stronger result proved, of which the Theorem is a corollary.} 
% - obscure; I mean the result they prove, from which they prove the thm.
Assaf derives an abstraction $\meansH{c}^\sharp$ for dependency from $\meansH{c}$ (for every $c$), by calculation, following Cousot~\cite{Cou99} and similar to data refinement by calculation~\cite{HeHoareSanders86,Morgan:dr}.
For this purpose and others, it is essential that loops be interpreted by fixpoint at the level of sets-of-sets, so standard fixpoint reasoning is applicable,  
as opposed to using $\dirimg{\meansT{\whilec{b}{c}}}$ which is not a fixed point per se.
In fact Assaf derives two abstract semantics $\meansH{c}^\sharp$:
one for dependency ($\NI$) and one that computes cardinality of low-variation, for quantitative information flow properties.  The cardinality abstraction is not in $k$-safety for any $k$.
%\dn{And cardinality bounds make sense for nondeterministic programs.}

Pasqua and Mastroeni have aims similar to Assaf \emph{et al}, and investigate several variations on set-of-set semantics of loops~\cite{MastroeniP18}.
Our example in Section~\ref{sec:htrans} is adapted from their work,
which uses examples to suggest that Assaf's definition (called ``mixed'' in~\cite{MastroeniP18}) is preferable.
They also point out that, for subset closed hyperproperties,
these variations are precise in the sense of our Theorem~\ref{thm:main},
strengthening the inequation in Assaf \emph{et al}.\footnote{Displayed formula  $\{ \meansT{c}T \mid T \in \mathbb{T} \} \subseteq \meansH{c}\mathbb{T}$ 
   following Theorem 1 of~\cite{AssafNSTT17}.}

A peculiarity of these works is the treatment of conditionals.  
Assaf \emph{et al} take the equation $\meansH{\ifc{b}{c}{d}} = \dirimg{\meansT{\ifc{b}{c}{d}}}$ as the definition of $\meansH{\ifc{b}{c}{d}}$, but this makes the definition of $\meansH{\args}$ non-compositional and thus the inductive proofs a little sketchy.  
%This author does not see 
We do not discern an explicit definition in~\cite{MastroeniP18},
but do find remarks like this: ``The definition of the collecting hypersemantics is just the additive lift\ldots for every statement, except for the while case.''  
As a description of fact, it is true for subset closed hyperproperties (see our Theorem~\ref{thm:main}). 
But it is unsuitable as a definition.
We show that a proper definition is possible.
%, and it also helps clarify the semantics of loops (Fig.~\ref{fig:powtransem}).

Non-compositionality for sequence and conditional seems difficult to reconcile with proofs by induction on program structure.
It also results in anomalies, e.g., the formulation in Assaf \emph{et al} means that in case $c$ is a loop, the semantics of $c$ 
is different from the semantics of $\ifc{true}{c}{\skipc}$.
The obscurity is rectified when the semantics is restricted to subset closed sets,
as spelled out in detail in an unpublished note~\cite{GotliboymNaumann},
confirming remarks by Pasqua and Mastroeni~\cite{MastroeniP17} 
elaborated in~\cite{MastroeniP18} and in the thesis of Pasqua.

Although similar to results in the preceding work, our results are novel in a couple of ways.
Our semantics is for a language with nondeterminacy, unlike theirs.
Of course, nondeterministic programs typically fail to satisfy $\NI$ and related properties,
and our theorem is restricted to deterministic programs.
A minor difference is that they formulate loop semantics in a standard form mentioned in Footnote~\ref{fn:tailrec} that asserts the negated guard following the fixpoint rather than as part of it.
Those works use semantics mapping traces to traces (or sets of sets thereof), 
like our $\meansT{\args}'$ and $\meansH{\args}'$, 
said to be needed in order to express dependency.  
We have shown that states-to-states is sufficient to exhibit both the anomaly and its resolution. 
It suffices for specifications of the form (\ref{eq:NI})
and may facilitate further investigation owing to its similarity to many variations on relational and transformer semantics.

% omittable
Apropos $\NI$, the formulation (\ref{eq:NI}) is robust in the sense that it generalizes to more nuanced notions of dependency:
$\dirimg{\dirimg{R}} \P \sbs \Q$ where $\P$ expresses agreement on some projections of the input (e.g., agreement on whether a password guess is correct, or agreement on some aggregate value derived from a sensitive database)
and $\Q$ expresses agreement on the observable output values.
Such policies are the subject of ~\cite{SabelfeldSandsDeclassJ}.

Banks and Jacob~\cite{BanksJ10a} formalize general confidentiality policies in UTP 
and introduce a family of confidentiality preserving refinement relations.
The ideas are developed further in subsequent work
where confidentiality-violating refinements are represented as miracles~\cite{BanksJ14}
and knowledge is explicitly represented by sets encoding alternate executions,
an idea that has appeared in other guises~\cite{Morgan09,AssafN16}.

%Backhouse $(r\to s) ; (t\to u) \subseteq (r;s) \to (t;u) $

%\url{https://www.isa-afp.org/entries/Transformer_Semantics.html}

%Elgi-Milner ordering for demonic nondeterminacy and termination~\cite{MollerS05}.

\section{Conclusion}\label{sec:concl}

%% vague Desiderata:

%% - express and manipulate specifications including hyperproperties, their equivalences and implications

%% - mix program and specification constructs for reasoning about designs

%% - retain familiar algebraic laws and also semantics of operators, hiding as much as possible the trace/hyper distinction 

We have given what, to the best of our knowledge, is the first 
compositional definition of semantics at the hyper level.
Moreover, we proved that it is the lift of a standard semantics when restricted to subset closed hyperproperties.  The latter is a ``forward collecting semantics''
in the terminology of abstract interpretation.
The new semantics includes nondeterministic constructs, although the lifting equivalence is only proved for the deterministic fragment.

Although deterministic noninterference is a motivating example, there are other interesting hyperproperties such as quantitative information flow that can be expressed in subset closed form  
and which are meaningful for nondeterministic programs.  
This is one motivation for further investigation including the following questions.
$\bullet$
Restricting to subset closed hyperproperties is sufficient to make possible a compositional fixpoint semantics at the hyper level that accurately represents the underlying semantics --- is it necessary?
$\bullet$ The h-transformer semantics allows nondeterministic choice and nondeterministic atoms that satisfy $\PSC$ (in light of Lemma~\ref{lem:meansHpsc}), and in fact the definitions 
can be used for a semantics in $\pset(\pset\Sta)\to\pset(\pset\Sta)$, into which 
$\Dpset(\pset\Sta)\to\Dpset(\pset\Sta)$ embeds nicely owing to joins being pointwise --- 
but what exactly is the significance of determinacy?  
$\bullet$ Are disjunctive transformers satisfying $\PSC$ closed under join?
What is a good characterization of transformers that are images of relations?

A person not familiar with unifying theories of programming may wonder whether programs are specifications.
Indeed, the author was once criticized by a famous computer scientist who objected to refinement calculi on the grounds that underspecification and nondeterminacy are distinct notions that ought not be confused --- though years later he published a soundness proof for a program logic, in which that confusion is exploited to good effect.
A positive answer to the question can be justified by embedding programs in a larger space of specifications, in a way that faithfully reflects a given semantics of programs.
Our theorem is a result of this kind.  The larger space makes it possible to express important requirements such as noninterference.
However, we do not put forward a compelling notion of specification that encompasses hyperproperties and supports a notion of refinement analogous to existing notions for trace properties.  Rather, we hope the paper inspires or annoys the reader enough to provoke further research. % in this direction.

%% \section{Themes and goals}

%% Specs as miraculous programs - can achieve by refusing to do, can choose the best angelically and compute the uncomputable.
%% But for a program to denote a class of programs, that class needs to be defined in some way like refinement.  Are there other ways, that encompass hypersafety? ssc? hyperliveness?  

%% Subset closure and its relevance to refinement and abstract interpretation.

%% Galois connections between different models (UTP, Cousot); or pocat embeddings.  

%% Models defined by healthiness conditions.  

%% Representing underspecification by nondeterminacy: some see as elegant, others see as confusion (Reynolds both); but it's limited to trace properties.

%% Goal: apply to heaplets (state as partial map from locations)  ---  can types be brought into this picture, for higher order, while keeping simple separation algebra?  Maybe naive model with variables but not stored pointers?  

%% Which laws hold at the hyper level? which are of interest? --- may want to express a functional spec conjoined with a confidentiality spec, and derive from there? 

\medskip

\textbf{Acknowledgements.}
Anonymous reviewers offered helpful suggestions 
and pointed out errors, omissions, and infelicities in an earlier version.

The authors were partially supported by NSF award 1718713.

\bibliographystyle{splncs04}
\bibliography{tony}

\end{document}

%%% Local Variables: 
%%% mode: latex
%%% TeX-master: t
%%% End: 